\documentclass[a4paper,12pt]{article}

\usepackage{mymacros}

\newcommand{\ket}[1]{\left| #1 \ra}
\newcommand{\Sch}{\cS ch}
\newcommand{\sm}{\mathrm{sm}}
\renewcommand{\ss}{\mathrm{ss}}

\newcommand{\sS}{{\mathsf{S}}}

\title{Quantum entanglement, Calabi-Yau manifolds,
and noncommutative algebraic geometry}
\author{Shinnosuke Okawa and Kazushi Ueda}
\date{}
\begin{document}

\maketitle

\begin{abstract}
We relate SLOCC equivalence classes of qudit states
to moduli spaces of Calabi-Yau manifolds
equipped with a collection of line bundles.
The cases of 3 qutrits and 4 qubits are also related
to noncommutative algebraic geometry.
\end{abstract}

\section{Introduction}
A \emph{qudit} is a $d$-level quantum mechanical system.
The $d=2$ case, called a \emph{qubit},
is the analogue of a classical bit,
and plays a fundamental role in quantum information theory.
An essential feature in quantum computing
is the existence of \emph{entanglement}.
A state
$\ket{\psi} \in V_1 \otimes \cdots \otimes V_n$
in an $n$-qudit system
is \emph{separable}
if there are states
$
 \ket{\psi^{(1)}} \in V_1, \ldots, \ket{\psi^{(n)}} \in V_n
$
such that
$
 \ket{\psi} = \ket{\psi^{(1)}} \otimes \cdots \otimes \ket{\psi^{(n)}}.
$
A state which is not separable is \emph{entangled}.

According to \cite{MR1804183},
two states $\ket{\psi}$ and $\ket{\psi'}$
are related by a sequence of
\emph{stochastic local operations and classical communication
(SLOCC)} if there are invertible operators
$
 A^{(1)} \in \GL(V_1), \ldots, A^{(n)} \in \GL(V_n)
$
such that
$
 \ket{\psi'} = (A^{(1)} \otimes \cdots \otimes A^{(n)}) \ket{\psi}.
$
Hence
the classification of quantum mechanical states up to SLOCC
is equivalent to the classification of
$\GL(V_1) \times \cdots \times \GL(V_n)$-orbits
in $V_1 \otimes \cdots \otimes V_n$.
Separable states constitute a single orbit,
and all other orbits are entangled.


In this paper, we relate the classification
of quantum mechanical states up to SLOCC
to the moduli space of algebraic varieties
with additional structures.
We will deal only with the case
$\dim V_1 = \cdots = \dim V_n = d$
for the sake of simplicity,
although the same idea works in greater generality.
\pref{th:main1} shows that
the moduli stack of $n$-qudit states is birational
to the moduli stack of pairs $(Y, (L_1, \ldots, L_{n-1}))$
of a Calabi-Yau manifold $Y$ of dimension $(n-1)(d-1)-d$
and a collection $(L_1, \ldots, L_{n-1})$
of line bundles on $Y$
satisfying some conditions.


Note that $Y$ is an elliptic curve
if and only if $(n, d) = (3,3)$ or $(4, 2)$.
Elliptic curves have already appeared
in several papers in quantum information theory,
such as
\cite{MR2039690,MR2105225,levay2011two}
to name a few.
%
In the case of $(n,d)=(3,3)$,
Theorems \ref{th:main1} and \ref{th:ell-triple} show that
general SLOCC equivalence classes of 3-qutrit states
are in one-to-one correspondence with
isomorphism classes of triples $(C, L_1, L_2)$,
where $C$ is an elliptic curve and
$L_1, L_2$ are line bundles of degree three on $C$.
This case is also related
to non-commutative projective planes,
which is a geometric incarnation
of three-dimensional  Sklyanin algebras
\cite{MR684124,Artin_Tate_Van_den_Bergh}.
In the case of $(n,d)=(4,2)$,
Theorems \ref{th:main1} and \ref{th:ell-quadruple} show that
general SLOCC equivalence classes of 4-qubit states
are in one-to-one correspondence with
isomorphism classes of quadruples $(C, L_1, L_2, L_3)$,
where $C$ is an elliptic curve and
$L_1, L_2, L_3$ are line bundles of degree two on $C$.
This case is related
to non-commutative quadric surfaces
\cite{MR2836401}.
%
These moduli spaces are also studied
in \cite{Bargava-Ho_CS}
from a different point of view.

As a higher-dimensional example,
consider the case $(n, d)=(5,2)$.
\pref{th:main1}
together with \cite[Theorem 7.1]{Bargava-Ho-Kumar_OP} shows that
general SLOCC equivalence classes of 5-qubit states
are in one-to-one correspondence
with isomorphism classes of $M$-polarized K3 surfaces,
where $M = \bZ e_1 \oplus \cdots \oplus \bZ e_4$
is a lattice of rank four
whose intersection matrix is given by
$(e_i, e_j) = 2 - 2 \delta_{ij}$.

The idea of using algebraic geometry
to classify SLOCC equivalence classes of entangled states
also appears in \cite{MR2246703, MR3050572}.
Classification of 4-qubit entanglement is initiated
in \cite{MR1910235}, and
related to string theory in \cite{MR2720335}.
See e.g. \cite{MR2994241} and references therein
for subsequent development.


\emph{Acknowledgment}.
We thank Keiji Oguiso
for pointing out the reference \cite{Bargava-Ho-Kumar_OP}.
S.~O. is supported by JSPS Grant-in-Aid for Young Scientists No.~25800017.
K.~U. is supported by JSPS Grant-in-Aid for Young Scientists No.~24740043.
A part of this work is done
while K.~U. is visiting Korea Institute for Advanced Study,
whose hospitality and nice working environment
is gratefully acknowledged.

\section{Moduli stacks}
 \label{sc:moduli}


For a pair $(d,n)$ of integers greater than one,
we define the category $\cM_{d,n}$ as follows:
\begin{itemize}
 \item
An object
$
 (\varphi\colon  \cY \to S, (\cL_1, \ldots, \cL_{n-1}))
$
consists of
\begin{itemize}
 \item
a smooth morphism $\varphi\colon  \cY \to S$
of schemes, and
 \item
a collection $(\cL_1, \ldots, \cL_{n-1})$ of line bundles on $\cY$
\end{itemize}
such that
\begin{itemize}
 \item
the relative canonical sheaf $\omega_{\cY/S}$ is trivial,
 \item
$\varphi_*(\cL_i)$ for $i=1, \ldots, n-1$ are
locally-free sheaves of rank $d$,
 \item
the composition morphism
\begin{align}
 \mu : \varphi_*(\cL_1) \otimes \cdots \otimes \varphi_*(\cL_{n-1})
  \to \varphi_*(\cL_1 \otimes \cdots \otimes \cL_{n-1})
\end{align}
is surjective with the kernel of rank $d$, and
 \item
the natural morphism
\begin{align}
 \cY
  \to \bP(\varphi_* \cL_1) \times_S \cdots \times_S \bP(\varphi_* \cL_{n-1})
\end{align}
is a closed embedding,
whose image is a complete intersection over $S$.
\end{itemize}
 \item
A morphism from
$
 (\varphi \colon \cY \to S, (\cL_1, \ldots, \cL_{n-1}))
$
to
$
 (\varphi' \colon \cY' \to S', (\cL_1', \ldots, \cL_{n-1}'))
$
is a commutative diagram
\begin{align}
\begin{CD}
 \cY @>{\phi}>> \cY' \\
  @V{\varphi}VV  @VV{\varphi'}V \\
  S @>{\phibar}>> S'
\end{CD}
\end{align}
such that
\begin{itemize}
 \item
$\phibar \colon S \to S'$ is a morphism of schemes,
 \item
$\phi \colon \cY \to \cY'$ induces an isomorphism
$\cY \simto S \times_{S'} \cY'$ of schemes, and 
 \item
$
 \phi^* \cL_i' \cong \cL_i
$
for any $i = 1, \ldots, n-1$.
\end{itemize}
\end{itemize}
The forgetful functor $\cM_{d,n} \to \Sch/\bC$
sending $(\varphi\colon  \cY \to S, (\cL_1, \ldots, \cL_{n-1}))$ to $S$
makes $\cM_{d,n}$ into a category fibered in groupoids.

Let $V_1, \ldots, V_n$ be vector spaces of dimension $d$.
The group
$
 G = \GL(V_1) \times \cdots \times \GL(V_n)
$
acts naturally on
$
 \cR = V_1 \otimes \cdots \otimes V_n,
$
and the quotient stack will be denoted by
\begin{align}
 \cQ_{d,n} = [\cR / G].
\end{align}
This is an Artin stack of finite type over $\bC$.
As a category, it is defined as follows:
\begin{itemize}
\item
An object $(S, \cP, \psi)$ consists of
\begin{itemize}
 \item
a scheme $S$,
 \item
a principal $G$-bundles $\pi \colon \cP \to S$, and
 \item
a $G$-equivariant morphism $\psi \colon \cP \to \cR$.
\end{itemize}

\item
A morphism $(\phi, \phibar)$ consists of
\begin{itemize}
 \item
a morphism $\phibar \colon S \to S'$ of schemes, and
 \item
an isomorphism $\phi : \cP \simto S \times_{S'} \cP'$ of principal $G$-bundles
\end{itemize}
such that the diagram
\begin{align}
\begin{CD}
 \cP @>{\psi}>> \cR \\
 @V{\phi}VV @AA{\pi'}A \\
 S \times_{S'} \cP' @>{\mathrm{pr}_2}>> \cP'
\end{CD}
\end{align}
is commutative.
\end{itemize}

\begin{theorem} \label{th:main1}
The category $\cM_{d,n}$ is an Artin stack
which is birational to $\cQ$.
\end{theorem}

\begin{proof}
Let $R = \sS(\cR)$ be the symmetric algebra over the vector space $\cR$.
For an element
$
 \eta \in \cR,
$
let
$
 V_\eta \subset V_1 \otimes \dots \otimes V_{n-1}
$
be the image of $\eta$
considered as an element of
$
 V_1 \otimes \cdots \otimes V_n
  \cong \Hom(V_n^\vee, V_1 \otimes \cdots \otimes V_{n-1}).
$
Let further
$
 I_\eta
  \subset \sS(V_1) \otimes \cdots \otimes \sS(V_{n-1})
$
be the ideal generated by $V_\eta$, and
$Y_\eta$ be the subscheme of
$\bP(V_1) \times \cdots \times \bP(V_{n-1})$
defined by $I_\eta$.
The open subscheme of $\cR$
consisting of elements $\eta$ such that
\begin{itemize}
 \item
$V_\eta$ is $d$-dimensional, and 
 \item
$Y_\eta$ is smooth,
\end{itemize}
will be denoted by $\cR^\sm$.
The $G$-action on $\cR$ preserves $\cR^\sm$, and
we write the corresponding quotient stack as
$
 \cQ^\sm = [\cR^\sm / G].
$
We will show an equivalence
\begin{align} \label{eq:main1}
 \cM_{d,n} \cong \cQ^\sm
\end{align}
of categories fibered in groupoids.

We first define a functor $\Psi\colon \cM_{d,n} \to \cQ^\sm$ as follows:
For an object 
$
 (\varphi \colon \cY \to S, (\cL_1, \ldots, \cL_{n-1}))
$
of the category $\cM_{d,n}$,
choose an open cover
$
 S=\bigcup_{\lambda \in \Lambda} S_\lambda
$
of $S$ such that
$
 \varphi_*(\cL_i)|_{S_\lambda}
$
is trivial for any $i = 1, \ldots, n-1$ and $\lambda \in \Lambda$.
Choose a trivialization
$
 \phi_{i, \lambda} \colon \varphi_*(\cL_i)|_{S_\lambda}
  \simto V_i \otimes \cO_{S_\lambda}
$
and use this to identify $\Ker(\mu)|_{S_\lambda}$
with the subsheaf of
$
 V_1 \otimes \cdots \otimes V_{n-1} \otimes \cO_{S_\lambda}.
$
We may assume that $S_\lambda$ is sufficiently small
so that $\Ker(\mu)|_{S_\lambda}$ is a free sheaf.
A choice of a trivialization
$
 \phi_n \colon \Ker(\mu)|_{S_\lambda} \simto V_n^\vee \otimes \cO_{S_\lambda}
$
gives a morphism $\psibar_\lambda \colon S_\lambda \to \cR$,
which corresponds the morphism of locally free sheaves
\[
\left(\otimes_{i=1}^{n-1}\varphi_{i,\lambda}\right)\circ\varphi_{n,\lambda}^{-1}\colon
V_n^\vee \otimes \cO_{S_\lambda}\hookrightarrow
V_1 \otimes \cdots \otimes V_{n-1} \otimes \cO_{S_\lambda}.
\]
For another open subscheme $S_\rho$,
the restrictions
$
 \psibar_\lambda |_{S_\lambda \cap S_\rho} \colon S_\lambda \cap S_\rho \to \cR
$
and
$
 \psibar_\rho |_{S_\lambda \cap S_\rho} \colon S_\lambda \cap S_\rho \to \cR
$
are related by a change of trivializations $\phi_i$ for $i=1, \ldots, n$.
Since a change of trivialization is given by the action
of the group $G = \GL(V_1) \times \cdots \times \GL(V_n)$,
one can form a principal $G$-bundle $\cP$ on $S$
by gluing $S_\lambda \times G$ by twisting by the action of $G$,
in such a way that $\psibar_\lambda$ lifts to a $G$-equivariant morphism
$\psi : \cP \to \cR$.
The smoothness of the morphism $\varphi\colon  \cY \to S$ implies
that the image of $\psi$ lies in $\cR^\sm$.
The action of $\Psi$ on morphisms is defined in the obvious way.

The functor $\Phi\colon \cQ^\sm \to \cM_{d,n}$ in the other direction
is defined as follows:
Let $\cY_{\cR^\sm}$ be the subscheme of
$\cR^\sm \times \bP(V_1) \times \cdots \times \bP(V_{n-1})$
whose fiber over $\eta \in \cR^\sm$ is
the subscheme $Y_\eta \subset \bP(V_1) \times \cdots \times \bP(V_{n-1})$.
For an object $(S, \cP, \psi)$ of $\cQ^\sm$,
the subscheme $\cY_\cP = \cY_{\cR^\sm} \times_{\cR^\sm} \cP$
of $\cP \times \bP(V_1) \times \cdots \times \bP(V_{n-1})$
has an action of $G$,
which makes the projection $\cY_\cP \to \cP$
into a $G$-equivariant morphism.
Since $G$-equivariance is the cocycle condition
with respect to the coequalizer diagram
\begin{align}
 \cP \times G \rightrightarrows \cP \to S,
\end{align}
the subscheme
$
 \cY_\cP \subset \cP \times \bP(V_1) \times \cdots \times \bP(V_{n-1})
$
descends to a subscheme
$
 \cY \subset S \times \bP(V_1) \times \cdots \times \bP(V_{n-1}).
$
Let $\varphi \colon \cY \to S$ be the projection to $S$.
The line bundles $\cL_i$ for $i=1, \ldots, n-1$ is defined
as the pull-back of the tautological line bundle $\cO_{\bP(V_i)}(1)$
on $\bP(V_i)$.
Then $(\varphi\colon  \cY \to S, (\cL_1, \ldots, \cL_{n-1}))$ gives an object of $\cM_{d,n}$.
The action of $\Phi$ on morphisms is defined in the obvious way.

The image of an object
$
 (\varphi \colon \cY \to S, (\cL_1, \ldots, \cL_{n-1}))
$
of $\cM_{d,n}$ by the functor $\Phi \circ \Psi$
is obtained as the zero of
\begin{align}
 \mu : \varphi_*(\cL_1) \otimes \cdots \otimes \varphi_*(\cL_{n-1})
  \to \varphi_*(\cL_1 \otimes \cdots \otimes \cL_{n-1})
\end{align}
in
$
 \bP(\varphi_*(\cL_1)) \otimes_S \cdots \otimes_S \bP(\varphi_*(\cL_{n-1})).
$
This is a complete intersection over $S$
of degree $(\bsone^d) = (\bsone, \ldots, \bsone)$,
where $\bsone = (1, \ldots, 1) \in \bZ^{n-1}$.
This contains the scheme $\cY$,
which must coincide with $\cY$
since $\cY$ is a complete intersection;
if $\cY$ is strictly smaller,
then $\cY$ cannot have the trivial canonical bundle
by the adjunction formula.
This shows that $\Phi \circ \Psi$ is isomorphic
to the identity functor.
The isomorphism of $\Psi \circ \Phi$
with the identity functor
follows immediately from the construction.
Hence $\Psi$ and $\Phi$ are quasi-inverse to each other,
and \pref{th:main1} is proved.
\end{proof}

Let
$
 H = \lc (\lambda_1 \id_{V_1}, \ldots, \lambda_n \id_{V_n})
  \in G \relmid \lambda_1 \cdots \lambda_n = 1 \rc
$
be the generic stabilizer
of the $G$-action on $\cR$, and
$G' = G/H$ be the quotient group
acting effectively on $\cR$.
Let further $S = R^{G'}$ be the invariant ring of $R$
with respect to the action of $G'$.
The ring $S$ inherits the natural grading
coming from the standard grading of $R$ as a symmetric algebra.
The \emph{geometric invariant theory (GIT) quotient}
of $\cR$ by $G$ is defined as $Q^\ss = \Proj S$
\cite{Mumford-Fogarty-Kirwan}.
A point $x \in \cR$ is \emph{semi-stable} if
there is a $G'$-invariant polynomial $s \in R^{G'}$ such that $s(x) \ne 0$.
Two $G$-orbits $O$ and $O'$ of $\cR^\ss$ are related
by the \emph{closure equivalence}
if the closures of $O$ and $O'$ intersect.
Geometric points of the GIT quotient are in one-to-one correspondence
with closure equivalence classes of $G$-orbits in $\cR^\ss$.
The GIT quotient $Q^\ss$ is the best approximation of the quotient stack
$
 \cQ^\ss = [\cR^\ss/G]
$
by a scheme.
The importance of GIT in quantum entanglement
is first pointed out by Klyachko
\cite{Klyachko_CSE}.

\begin{proposition} \label{pr:main2}
The moduli stack $\cM_{d,n}$
is birational to $\cQ^\ss$.
\end{proposition}

\begin{proof}
Since the hyperdeterminant of format $2 \times \cdots \times 2$
is non-trivial and $G'$-invariant
by \cite[Theorems 14.1.3 and 14.1.4]{Gelfand-Kapranov-Zelevinsky_DRMD},
the semi-stable locus $\cR^\ss$ is non-empty.
Since semi-stability is an open condition,
$\cR^\ss$ is a non-empty open subscheme of $\cR$.
The theorem of Bertini
shows that $\cR^\sm$ is a non-empty open subset of $\cR$.
Hence $\cR^\ss$ and $\cR^\sm$ has a common open dense subset,
so that the stacks $\cM_{d,n} \cong [\cR^\sm / G]$
and $\cQ^\ss=[\cR^\ss/G]$ are birational.
\end{proof}

Let $[\cR^\ss/G']$ be the stack obtained from $\cQ^\ss$
by removing the generic stabilizer.
The following proposition shows that
the forgetful morphism from $\cM_{d,n}$
to the moduli space of Calabi-Yau manifolds
(without any additional structure) is dominant
on an irreducible component:

\begin{lemma}
One has $\dim H^1(T_Y) = \dim [\cR^\ss/G']$
if $\dim Y \ge 3$.
\end{lemma}

\begin{proof}
Let 
$
 \pi_i : \bP = \bP(V_1) \times \cdots \times \bP(V_{n-1})
  \to \bP(V_i)
$
be the $i$-th projection
and set
\begin{align*}
 \cO_\bP(\bsa)
  &= \cO_{\bP(V_1)}(a_1) \boxtimes \cdots
   \boxtimes \cP_{\bP(V_{n-1})}(a_{n-1}) \\
  &= \bigotimes_{i=1}^{n-1} \pi_i^* \cO_{\bP(V_i)}(a_i)
\end{align*}
for
$
 \bsa = (a_1,\ldots,a_{n-1}) \in \bZ^{n-1}.
$
Let $\{ \bse_i \}_{i=1}^{n-1}$ be the standard basis
of $\bZ^{n-1}$ and set $\bsone = e_1 + \cdots + e_{n-1}$.
In the exact sequence
\begin{align} \label{eq:normal}
 0 \to \cT_Y \to \cT_\bP|_Y
  \to \cN_{Y/\bP}
  \to 0,
\end{align}
one has
$
 \cN_{Y / \bP} \cong \cO_Y(\bsone)^{\oplus d}
$
since $Y$ is a complete intersection
of degree $(\bsone^d)$.
By tensoring the Koszul resolution
\begin{multline} \label{eq:Kozsul_Y}
 0
  \to\cO_\bP(-d \cdot \bsone)
  \to \cdots
  \to \cO_\bP(-2 \cdot \bsone)^{\oplus \binom{d}{2}}
  \to \cO_\bP(-\bsone)^{\oplus d}
  \to \cO_\bP \to \cO_Y \to 0
\end{multline}
with $\cO_Y(\bsone)$
and taking the global section,
one obtains
\begin{align}
 \dim H^0(\cO_Y(\bsone)) = d^{n-1}-d,
\end{align}
so that
\begin{align}
 \dim H^0(\cN_{Y/\bP})
  = \dim H^0(\cO_Y(\bsone))^{\oplus d}
  = d(d^{n-1}-d).
\end{align}
The pull-back of the Euler sequence on $\bP(V_i)$ gives
\begin{align}
 0 \to \cO_\bP
  \to \cO_\bP(\bse_i)^{\oplus d}
  \to \pi_i^* \cT_{\bP(V_i)}
  \to 0,
\end{align}
which together with
$
 \cT_\bP = \bigoplus_{i=1}^{n-1} \pi_i^* \cT_{\bP(V_i)}
$
gives
\begin{align}
 0 \to \cO_\bP^{\oplus (n-1)}
  \to \bigoplus_{i=1}^{n-1} \cO_\bP(\bse_i)^{\oplus d}
  \to \cT_\bP
  \to 0.
\end{align}
By restricting to $Y$, one obtains
\begin{align}
 0 \to \cO_Y^{\oplus (n-1)}
  \to \bigoplus_{i=1}^{n-1} \cO_Y(\bse_i)^{\oplus d}
  \to \cT_\bP|_Y
  \to 0.
\end{align}
One can show
from \eqref{eq:Kozsul_Y}
that $H^1(\cO_Y(\bse_i))=0$,
so that
\begin{align}
 \dim H^0(\cT_\bP|_Y)
  &= (n-1)d^2 - (n-1).
\end{align}
and
\begin{align}
 H^1(\cT_\bP|_Y) \cong H^2(\cO_Y^{\oplus (n-1)}).
\end{align}
One has
$
 H^0(\cT_Y) = 0
$
if $\dim Y \ge 2$,
and
$
 H^2(\cO_Y) = 0
$
if $\dim Y \ge 3$.
The long exact sequence
associated with \eqref{eq:normal} gives
\begin{align}
 0
  \to H^0(\cT_\bP|_Y)
  \to H^0(\cN_{Y/\bP})
  \to H^1(\cT_Y)
  \to 0.
\end{align}
It follows that
\begin{align*}
 \dim H^1(\cT_Y)
  &= \dim H^0(\cN_{Y/\bP}) - \dim H^0(\cT_\bP|_Y) \\
  &= d(d^{n-1}-d) - ((n-1)d^2-(n-1)) \\
  &= d^n-n d^2 + n - 1,
\end{align*}
which coincides with
\begin{align*}
 \dim [\cR^\ss/G']
  &= \dim \cR - \dim G' \\
  &= d^n - (n d^2-n+1) \\
  &= d^n-n d^2+n-1.
\end{align*}
\end{proof}

\section{Noncommutative algebraic geometry}
 \label{sc:ncag}

A $\bZ$-algebra is an algebra of the form
$
 A = \bigoplus_{i,j \in \bZ} A_{ij}
$
with the property that the product satisfies
$A_{ij} A_{kl}= 0$ for $j \ne k$ and
$
 A_{ij} A_{jk} \subset A_{ik}.
$
We assume that $A_{ii}$ for any $i \in \bZ$
has an element $e_i$, called the \emph{local unit},
which satisfies $f e_i = f$ and $e_i g = g$
for any $f \in A_{ki}$ and $g \in A_{ij}$.
A graded $A$-module
is a right $A$-module of the form
$M = \bigoplus_{i \in \bZ} M_i$
such that $M_i A_{kl} = 0$ for $k \ne i$ and
$M_i A_{ij} \subset A_j$.
An $A$-module is a \emph{torsion}
if it is a colimit of modules $M$
satisfying $M_i = 0$ for sufficiently large $i$.
Let $P_i = e_i A$ and $S_i = e_i A e_i$
be right $A$-modules.
A $\bZ$-algebra $A$ is \emph{positively graded}
if $A_{ij} = 0$ for $i < j$.
A positively graded $\bZ$-algebra $A$ is \emph{connected}
if $\dim A_{ij} < \infty$ and
$A_{ii} = \bC \, e_i$ for any $i, j \in \bZ$.
A connected $\bZ$-algebra $A$ is \emph{AS-regular} if
\begin{itemize}
 \item
$\dim A_{ij}$ is bounded by a polynomial in $j-i$,
 \item
the projective dimension of $S_i$ is bounded
by a constant independent of $i$,
 \item
$\sum_{j,k} \dim \Ext_{\Gr A}^j(S_k, P_i) = 1$
for any $i \in \bZ$.
\end{itemize}
An AS-regular $\bZ$-algebra $A$ is
a \emph{3-dimensional quadratic AS-regular $\bZ$-algebra}
if the minimal resolution of $S_i$ is of the form
\begin{align}
 0 \to P_{i+3} \to P_{i+2}^3 \to P_{i+1}^3 \to P_i \to S_i \to 0.
\end{align}
An AS-regular $\bZ$-algebra $A$ is
a \emph{3-dimensional cubic AS-regular $\bZ$-algebra}
if the minimal resolution of $S_i$ is of the form
\begin{align}
 0 \to P_{i+4} \to P_{i+3}^2 \to P_{i+1}^2 \to P_i \to S_i \to 0.
\end{align}
Let
$
 \Qgr A = \Gr A / \Tor A
$
be the quotient abelian category
of the abelian category $\Gr A$
of graded $A$-modules
by the Serre subcategory $\Tor A$
consisting of torsion modules.
A \emph{noncommutative projective plane}
is an abelian category of the form $\Qgr A$
for a 3-dimensional quadratic AS-regular $\bZ$-algebra.
A \emph{noncommutative quadric surface} is defined similarly
as $\Qgr A$ for a 3-dimensional cubic AS-regular $\bZ$-algebra
\cite[Definition 3.2]{MR2836401}.

The classification of 3-dimensional regular quadratic $\bZ$-algebras
can be found in \cite[Proposition 3.3]{MR2836401},
which is essentially due to \cite{MR1230966}.
They are divided into the \emph{linear} case and the \emph{elliptic} case.
The abelian category $\Qgr A$ is equivalent
to $\Qcoh \bP^2$ in the linear case.
The elliptic cases are classified by triples
$(C, L_1, L_2)$ of a curve $C$
and a pair $(L_1, L_2)$ of line bundles such that
\begin{itemize}
 \item
the curve $C$ is embedded as a divisor of degree 3
in $\bP^2$ by global sections of both $L_1$ and $L_2$,
 \item
$\deg(L_1|_D) = \deg(L_2|_D)$
for every irreducible component $D$ of $C$, and
 \item
$L_1$ is not isomorphic to $L_2$.
\end{itemize}
Accordingly,
the moduli stack $\cN_3$ of elliptic triples
is defined as a category
whose object $(\varphi\colon  \cY \to S, (\cL_1, \cL_2))$
consists of a flat morphism $\varphi\colon  \cY \to S$ of schemes
and a pair $(\cL_1, \cL_2)$ of line bundles on $\cY$
satisfying the above three conditions
on geometric points of $S$.
The morphisms are defined
in the same way as that for $\cM_{d,n}$.

\begin{theorem} \label{th:ell-triple}
The moduli stack $\cM_{3,3}$ is birational
to the moduli stack $\cN_3$.
\end{theorem}

\begin{proof}
Let $L$ be a line bundle on an elliptic curve.
Serre duality
\begin{align}
 H^1(L)^\vee \cong H^0(L^\vee)
\end{align}
shows that $h^1(L) = 0$ if $h^0(L) > 1$, and
Riemann-Roch theorem
\begin{align}
 h^0(L) - h^1(L) = \deg L
\end{align}
shows that $\deg L=3$ if and only if $h^0(L) = 3$.
If this is the case, then $L$ is very ample (%
cf.~e.g.~\cite[Corollary IV.3.2(b)]{Hartshorne}).
It follows that the open substack of $\cM_{3,3}$
consisting of objects $(\varphi\colon  \cY \to S, (\cL_1, \cL_2))$
such that $\cL_1$ is not isomorphic
to the tensor product of $\cL_2$ and a pull-back
of a line bundle on $S$
is isomorphic to the open substack of $\cN_3$
consisting of objects $(\varphi\colon  \cY \to S, (\cL_1, \cL_2))$
such that $\varphi$ is smooth,
and \pref{th:ell-triple} is proved.
\end{proof}


The classification of 3-dimensional cubic AS-regular $\bZ$-algebras
is given in \cite[Proposition 4.2]{MR2836401}.
They are divided into the \emph{linear} case and the \emph{elliptic} case.
The abelian category $\Qgr A$ is equivalent
to $\Qcoh \bP^1 \times \bP^1$ in the linear case.
The elliptic cases are classified by quadruples
$(C, L_1, L_2, L_3)$ of a curve $C$
and line bundles $L_1$, $L_2$, and $L_3$ such that
\begin{itemize}
 \item
the curve $C$ is embedded as a divisor of bidegree $(2,2)$
in $\bP^1 \times \bP^1$
by global sections of both $(L_1, L_2)$
and $(L_2, L_3)$,
 \item
$\deg(L_1|_D) = \deg(L_3|_D)$
for every irreducible component $D$ of $C$, and
 \item
$L_1$ is not isomorphic to $L_3$.
\end{itemize}
The moduli stack $\cN_4$ of elliptic quadruples
is defined as a category
whose object $(\varphi\colon  \cY \to S, (\cL_1, \cL_2, \cL_3))$
consists of a flat morphism $\varphi\colon  \cY \to S$ of schemes
and a triple $(\cL_1, \cL_2, \cL_3)$ of line bundles on $\cY$
satisfying the above three conditions
on geometric point of $S$.

\begin{lemma} \label{lm:segre}
Let $(L, L')$ be a pair of line bundles of degree 2
on an elliptic curve $C$.
Then the morphism
$\varphi_{(L,L')} : C \to \bP(H^0(L)) \times \bP(H^0(L'))$
defined by $(L, L')$ is an embedding
if and only if $L$ is not isomorphic to $L'$.
\end{lemma}

\begin{proof}
It is clear that $\varphi_{(L,L)}$
factors through the diagonal embedding
\begin{align}
 \bP(H^0(L)) \to \bP(H^0(L)) \times \bP(H^0(L)),
\end{align}
and hence can never be an embedding.
Consider the composition map
\begin{align}
 \mu : H^0(L) \otimes H^0(L')
  \to H^0(L \otimes L').
\end{align}
Note that $h^0(L) = h^0(L') = 2$ and
$h^0(L \otimes L') = 4$.
If the map $\mu$ is surjective,
then the composition of the morphisms
\begin{align}
 C
  \xto{\ \varphi_{(L,L')}\ } \bP(H^0(L)) \times \bP(H^0(L'))
  \xto{\text{ Segre }} \bP(H^0(L) \otimes H^0(L'))
  \xto{\ \mu_*\ } \bP(H^0(L \otimes L'))
\end{align}
gives the morphism
$
 \varphi_{L \otimes L'} : C \to \bP(H^0(L \otimes L'))
$
associated with the line bundle $L \otimes L'$.
Here, the second arrow is the Segre embedding,
and the third arrow is induced by $\mu$.
Note that the line bundle $L \otimes L'$ has degree 4,
and hence very ample,
so that $\varphi_{L \otimes L'}$ is an embedding.
The fact that this embedding factors
through $\varphi_{(L,L')}$ shows that
the morphism $\varphi_{(L,L')}$ is an embedding.

Assume that the map $\mu$ is not surjective.
Choose a basis $\{ s, t \}$ of $H^0(L)$.
If one takes an element of $\ker \mu$,
then it can be written as
$
 s \otimes a + t \otimes b \in \ker \mu
$
for some $a, b \in H^0(L')$.
It follows that $s a = - t b$ in $H^0(L \otimes L')$,
which implies that $s$ divides $b$.
Since $\deg L = \deg L'$,
this is possible only if $L$ and $L'$ are isomorphic,
and \pref{lm:segre} is proved.
\end{proof}

\begin{theorem} \label{th:ell-quadruple}
The moduli stack $\cM_{4,2}$ is birational
to the moduli stack $\cN_4$.
\end{theorem}

\begin{proof}
Adjunction formula shows that
an elliptic curve embedded in $\bP^1 \times \bP^1$
must be a divisor of bidegree $(2,2)$.
\pref{lm:segre} shows that
the open substack of $\cM_{4,2}$
consisting of objects $(\varphi\colon  \cY \to S, (\cL_1, \cL_2, \cL_3))$
such that $\cL_1$, $\cL_2$ and $\cL_3$ are distinct
is isomorphic to the open substack of $\cN_4$
consisting of objects $(\varphi\colon  \cY \to S, (\cL_1, \cL_2, \cL_3))$
such that $\varphi$ is smooth.
\end{proof}

The geometry of the moduli stacks
$\cM_{3,3}$, $\cM_{4,2}$, $\cN_3$, $\cN_4$
and the corresponding compact moduli schemes
will be studied in more detail in
\cite{Abdelgadir-Okawa-Ueda_ncP2, Okawa-Ueda_ncquadric}
from the viewpoint of noncommutative algebraic geometry.

%
%
%
%
%

\bibliographystyle{amsalpha}
\bibliography{bibs}

\noindent
Shinnosuke Okawa

Department of Mathematics,
Graduate School of Science,
Osaka University,
Machikaneyama 1-1,
Toyonaka,
Osaka,
560-0043,
Japan.

{\em e-mail address}\ : \  okawa@math.sci.osaka-u.ac.jp
\ \vspace{0mm} \\

\ \vspace{0mm} \\

\noindent
Kazushi Ueda

Department of Mathematics,
Graduate School of Science,
Osaka University,
Machikaneyama 1-1,
Toyonaka,
Osaka,
560-0043,
Japan.

{\em e-mail address}\ : \  kazushi@math.sci.osaka-u.ac.jp
\ \vspace{0mm} \\

\end{document}